\documentclass[runningheads]{llncs}
\usepackage[paperheight=235mm, paperwidth=155mm,textwidth=12.2cm,textheight=19.3cm]{geometry}
\usepackage{bbding}

\usepackage{amssymb,amsmath,amsthm}
\usepackage{mathtools}
\usepackage{stmaryrd}

\usepackage{algorithm}
\usepackage{algpseudocode}

\usepackage{url}
\AtBeginDocument{%
  \def\doi#1{\url{https://doi.org/#1}}}
\makeatother

\usepackage{tcolorbox}

\usepackage[]{todonotes}

\usepackage{float}
\usepackage[caption=false]{subfig}

\usepackage{tikz}
\usetikzlibrary{matrix}
\usetikzlibrary{arrows}
\usetikzlibrary{shapes}
\usetikzlibrary{automata, positioning,through,calc}

\tikzset{
->, 
>=stealth, 
node distance=2cm, 
every state/.style={thick, fill=gray!10,ellipse}, 
interm state/.style={thick, fill=gray!10, draw, rectangle, inner sep=6}, 
initial text=$ $, 
}

\newcommand{\lat}{\mathrm{lat}}
\newcommand{\fl}{\mathrm{flow}}

\newcommand{\SC}{\mathrm{SC}}
\newcommand{\canflow}{\sqsubseteq}
\newcommand{\SCTop}{\top}
\newcommand{\SCBot}{\bot}
\newcommand{\Label}{L}

\newcommand{\toLattice}[1]{{\llbracket #1 \rrbracket}_{\lat}}
\newcommand{\toFlowRel}[1]{{\llbracket #1 \rrbracket}_{\fl}}

\newcommand{\toLatC}[1]{{\llbracket #1 \rrbracket}}
\newcommand{\auxClasses}{\mathrm{addLeastUpperLabels}}

\newcommand{\InVars}{\mathrm{X}}
\newcommand{\OutVars}{\mathrm{Y}}
\newcommand{\AllVars}{\mathrm{Z}}

\newcommand{\AFlow}{A_{\fl}}
\newcommand{\GFlow}{G_{\fl}}

\newcommand{\ContA}{A}
\newcommand{\ContG}{G}

\newcommand{\Flows}{\mathcal{M}}

\newcommand{\FlowsFrom}{\mathrm{U}}
\newcommand{\FlowsTo}{\mathrm{V}}

\newcommand{\tAnd}{\text{ and }}

\newcommand{\dwFs}{distw\_f\_s}
\newcommand{\dwBs}{distw\_b\_s}
\newcommand{\ods}{wheel\_tick}
\newcommand{\dwFt}{distw\_f\_t}
\newcommand{\dwBt}{distw\_b\_t}
\newcommand{\odt}{odometer}

\begin{document}

\title{Information-flow Interfaces and\\Security Lattices} 
 
%
\author{
Ezio Bartocci\inst{1} \and
Thomas A. Henzinger\inst{2}\and \\
Dejan Nickovic\inst{3} \and
Ana Oliveira da Costa\inst{2}(\Envelope)
 }
\authorrunning{E. Bartocci et al.}
%
%
 \institute{
 Technische Universit\"at Wien, Vienna, Austria \\
 \email{ezio.bartocci@tuwien.ac.at}  
 \and
IST Austria, Klosterneuburg, Austria \\ \email{\{tah, ana.costa\}@ist.ac.at} \and
AIT Austrian Institute of Technology, Vienna, Austria \\ \email{dejan.nickovic@ait.ac.at}
}

\maketitle

\begin{abstract}
Information-flow interfaces is a formalism recently proposed for specifying, composing, and refining system-wide security requirements. 
In this work, we show  how the widely used concept of security lattices provides a natural semantic interpretation for information-flow interfaces.
\end{abstract}

\section{Introduction}

Modern information and communication technologies are reaching unprecedented size and complexity, exposing them more and more to a wide variety of cyber-attacks. The security-by-design engineering approach addresses this problem by enforcing security requirements throughout all phases of the design cycle, including early design stages.

\emph{Information-flow interfaces}~\cite{IFInterfaces22FASE,IFInterfaces20Arxiv} is a recently introduced security-by-design formalism for the compositional development of systems that are guaranteed to implement specified information-flow policies. 
This formalism defines information-flow requirements using \emph{no-flow} relations. 
A no-flow relation specifies a forbidden exchange of information from one system variable to another. 
This framework allows the engineer to combine \emph{top-down} and \emph{bottom-up} design activities: a top-down step consists of decomposing the system-wide requirements and mapping them to sub-systems and components, while a bottom-up step consists of assembling the system by combining available elements. 
To support the compositional design of systems, information-flow interfaces distinguish between \emph{assumptions} about the component's environment and \emph{guarantees} that the component provides when it operates in a proper environment. 
In addition, information-flow interfaces satisfy the two main properties of an interface theory: \emph{incremental design} and \emph{independent implementability}. 
Incremental design enables composing sub-systems and components without knowing the complete design context. 
Independent implementability allows the development of sub-systems by different design teams with system integration that guarantees the preservation of system-level requirements.

We illustrate the top-down design using information-flow interfaces with an automotive example -- a simplified version of a shared communication infrastructure connecting a wheel sensor and distance warners to the braking system and the odometer~\cite{IFInterfaces22FASE}. 
The shared communication infrastructure consists of (1) two (front and back) distance warners, which are sensors that estimate the proximity of the vehicle to other objects; (2) the wheel sensor, which senses the wheel rotations; (3) the odometer, a component that uses the wheel sensor data to measure the distance travelled by a vehicle; (4) a braking system, which takes the decision when to brake based on the information from the distance warners and (5) a shared bus that enables the communication between the other components. 
In this example, the braking system has a safety-critical function; consequently, communication with the distance warners has high-integrity. 
In contrast, the odometer and wheel sensor communication has low-integrity requirement. 
The main system-level requirement for this system is to guarantee the integrity of the communication channel when performing the safety-critical function. 
In other words, the design forbids any flow of information from the wheel sensor to the braking system.

\begin{figure}[htb]
\centering
\scalebox{0.7}{ 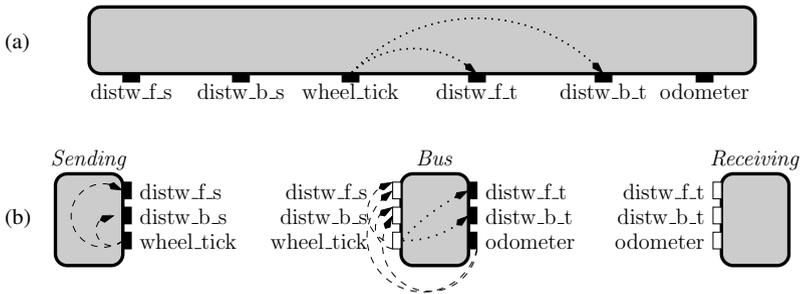 }
\caption{An example of an automotive shared communication infrastructure specified with information-flow interfaces.}
\label{fig:overview}
\end{figure}

Figure~\ref{fig:overview} shows (a) the system-level interface specifying the above information-flow requirements and (b) one decomposition and refinement step. 
We depict an interface with a gray rounded-edged box. 
The white and black rectangles attached to the interface denote its input and output ports, respectively. 
We use dashed arrows to specify no-flow relations. 
A no-flow relation is an assumption if its target is an input port and a guarantee otherwise.
The system-level security requirement from Figure~\ref{fig:overview}~(a) consists of two no-flow relations that forbid the exchange of information between the wheel sensor and the front and back distance warner targets. 
We observe that this interface specifies the properties of a closed system without defining its components and their interaction.

In the refinement step, shown in Figure~\ref{fig:overview}~(b), the interface is decomposed into three sub-systems: sending, receiving and bus. The two system-level requirements are naturally transferred to the bus sub-system. However, mapping these two properties from the system-level interface to the bus interface is not sufficient to enforce the original requirements. Without additional restrictions, there would be for instance an allowed flow in the decomposed system from the wheel sensor to the odometer, from the odometer to the distance wheeler sources, and from the sources to the distance wheeler targets, resulting in the violation of the system level requirements. To guarantee the enforcement of the system-level requirement after integration of sub-systems, we need to add additional assumptions (no-flow relations from the odometer and the wheel sensor to the distance wheelers in the bus sub-system) and guarantees (no-flow relations from the wheel sensor to the distance wheelers). In the subsequent design step, these three interfaces can be either further decomposed into smaller components, or assigned to different design teams for implementation.  

Information-flow interfaces follow a declarative style of specification with the focus on what are the forbidden flows of information between system variables, rather then on how to guarantee the absence of the specified flows. This is a design choice that gives a syntactic flavor to the interface theory and its operations -- composition of interfaces and their refinement consist of manipulating no-flow relations only. The advantage of this approach is that it enables simple incremental development of secure system architectures, while allowing to postpone semantic and implementation choices for the components to the later design stages. Information-flow interfaces have been equipped with a semantic interpretation in terms of \emph{contracts} defined as pairs of assumptions and guarantees being sets of \emph{flow relations}~\cite{IFInterfaces20Arxiv}.

In this paper, we propose \emph{security lattices} as an alternative semantic interpretation for information-flow interfaces. 
Equipping this interface theory with semantics based on security lattices has two major advantages. 
First, security lattices are commonly used in the security community for specifying information-flow and other security-related policies \cite{latticecontrol93,sabelfeld2003language,focardi2011types,austin2012multiple,IFCIneficient20}. 
Therefore, their integration into the interface theory facilitates the security engineers' adoption of the compositional design framework. 
Second, the information-flow theory's compositional nature can help improve the efficiency of verification mechanisms whose performance depends on the shape of the lattice (e.g., number of labels, number of interactions between labels, ...).  
By using the information-flow interfaces equipped with security lattices, we can help to design more efficient lattices by (i) decomposing the system design into sub-systems and removing unnecessary labels and (ii) providing a list of lattices to choose from.

\section{Background}

This section provides the necessary background for equipping the information-flow interface theory with a semantic interpretation that uses security lattices. We first recall information-flow contracts as the original semantic interpretation for information-flow interfaces~\cite{IFInterfaces20Arxiv}. We then provide an overview of security lattices~\cite{denninglattice76} for defining security policies. Throughout this section, we use the shared communication infrastructure example from Figure~\ref{fig:overview} to illustrate the semantic interpretation of interfaces using information-flow contracts and sketch the idea of replacing contracts with more common security lattices.  

\subsubsection{Information-Flow Contracts.}

The information-flow interfaces use the syntax described in the shared communication infrastructure example and illustrated in Figure~\ref{fig:overview} to specify security requirements in terms of no-flow relations. \emph{Information-flow contracts} \cite{IFInterfaces20Arxiv} provide a semantic interpretation to the interface theory. They define assumptions on the environment and guarantees on implementations in terms of \emph{flow relations}.
A relation \(\Flows \subseteq (\FlowsFrom\mathop{\cup} \FlowsTo) \times \FlowsTo\) is a flow relation iff it is a transitive relation over \(\FlowsFrom\mathop{\cup} \FlowsTo\), and reflexive over~\(\FlowsTo\).
Let \(\InVars\) and \(\OutVars\) be disjoint sets of input and output variables, respectively, with the set of all variables defined as \(\AllVars = \InVars \mathop{\cup} \OutVars\).
An \emph{information-flow contract} is a tuple \((\InVars, \OutVars, \AFlow, \GFlow)\) where \(\AFlow\subseteq 2^{\AllVars \times \InVars}\) is a set of flow relations to input variables, called \emph{(contract) assumption}; and \(\GFlow\subseteq 2^{\AllVars \times \OutVars}\) is a set of flow relations to output variables, called \emph{(contract) guarantee}.

\begin{example}
\begin{figure}
\centering
\scalebox{0.48}{ 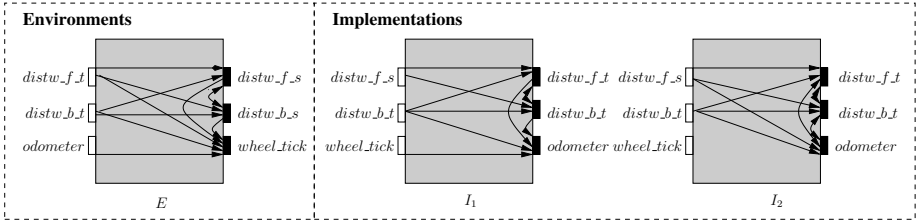 }
\caption{Information-flow contract for the shared communication infrastructure interface, where $E$ is the maximal permissible environment, while $I_1$ and $I_2$ are maximal implementations that satisfy the interface specification.}
\label{fig:contract}
\end{figure}

We refer to the bus sub-system in the shared communication infrastructure from Figure~\ref{fig:overview} to illustrate its semantic interpretation with information-flow contracts. 
Figure~\ref{fig:contract} depicts the information-flow contract for this interface. 
We use sets of components to graphically portray contract assumption and guarantee, where we depict a component as a gray box decorated with input and output variables, with solid arrows representing allowed flows. 
The assumption contains the set of all permissible environments, while the guarantee contains the set of all correct implementations. 
In this example, we depict only the maximal components, i.e., the contract assumption and guarantee contain all components that refine these components.
In particular, there is one maximal environment $E$ and two maximal implementations  $I_1$ and $I_2$.
\end{example}

\subsubsection{Security Lattices.}

Information-flow policies are usually defined for a set of \emph{(security) labels}, also referred to in the literature as security classes.
Labels and objects of a system (for example, variables or ports) usually define distinct sets of objects, with a label categorizing the role of an object within an information-flow policy.

A policy is specified with a \emph{label can-flow} relation \(\sqsubseteq\) over a set of security labels \(\SC\).
Then, \(\Label \sqsubseteq \Label'\), for labels \(\Label \in \SC\) and \(\Label' \in \SC\), means that information on entities assigned to the label \(\Label\) is allowed to flow to entities in \(\Label'\). 
Note that flow relations are defined for system's objects while can-flow relations are defined over security labels.
Additionally, a policy provides a joint operator \(\oplus \subseteq \SC \times\SC\) specifying how to combine entities assigned with different labels.
This operator can be used to dynamically assign a security label to objects that had no prior assignment.
A policy is fully specified by a tuple  \((\SC, \sqsubseteq, \oplus)\).
In \cite{denninglattice76}, Denning proposed the following axioms for security policies \((\SC, \sqsubseteq, \oplus)\):
\begin{description}
\item [Finiteness] the set of security classes \(\SC\) is finite;
\item [Order] \((\SC,\sqsubseteq)\) defines a partial order; 
\item [Public Label] there exists an unique lower bound in \(\SC\) with respect to \(\sqsubseteq\);
\item [Totally of Label Combining] the joint operator \(\oplus\) is a least upper bound operator defined for every pair of security labels.
\end{description}
A security policy that satisfies these assumptions, 
defines a bounded lattice \cite{denninglattice76}.
Then, \(\sqsubseteq\) is a reflexive, transitive, and antisymmetric relation over \(\SC\) with a unique upper bound (represented by \(\SCTop\)) and a unique lower bound (represented by~\(\SCBot\)).
From now on, for simplicity of presentation, we denote security policies by \((\SC,\canflow)\) only. 
Whenever \((\SC,\canflow)\) defines a lattice, the definition of \(\oplus\) over a pair of labels in \(\SC\) is given by their least upper bound in \((\SC,\canflow)\).

\begin{example}
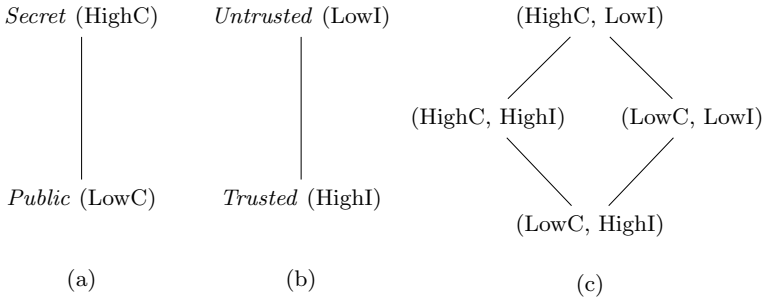
\begin{figure}
    \centering
\scalebox{0.95}{
\begin{tikzpicture}[node distance = 2cm]
	\node (topC) {\emph{Secret} (HighC)};
	\node [below = of topC] (belowC)  {\emph{Public} (LowC)};
	\node [below = 0.6 of belowC] (labelC)  {(a)};

	\node [right= 0.5cm of topC] (topI) {\emph{Untrusted} (LowI)};
	\node [below = of topI] (belowI)  {\emph{Trusted} (HighI)};
	\node [below = 0.6 of belowI] (labelI)  {(b)};
	
	\node [right= 1.5cm of topI] (top2) {(HighC, LowI)};
	\node [below left of = top2] (right2)  {(HighC, HighI)};
	\node [below right of = top2] (left2)  {(LowC, LowI)};
	\node [below = 1.5cm of $(right2.north)!0.5!(left2.north)$] (bottom2) {(LowC, HighI)};
	\node [below = 0.3 of bottom2] (labelC)  {(c)};
	
	\draw[-] (topC)--(belowC);
	
	\draw[-] (topI)--(belowI);

	\draw[-] (top2)--(right2);
	\draw[-] (top2)--(left2);
	\draw[-] (left2)--(bottom2);
	\draw[-] (right2)--(bottom2);
 	
\end{tikzpicture}
}
\caption{Security lattices for (a) confidentiality, (b) integrity and (c) both combined.}
\label{fig:enter-label}
\end{figure}

In Figure \ref{fig:enter-label}, we have three examples of security lattices represented as Hasse diagrams.
Reflexive and transitive arrows are omitted in the diagrams.
The two Hasse diagrams to the left depict two linear lattices: one for confidentiality and one for integrity. 
As illustrated by the lattices, confidentiality and integrity adopt inverse views on the flow from high to low labelled variables.
The lattice on the right side depicts the result of combining the confidentiality and integrity lattices.
%
%
Figure~\ref{fig:sl} below depicts a security lattice that provides an elegant and succinct representation of the maximal implementation $I_1$ from Figure~\ref{fig:contract}.
\begin{figure}
\centering
\scalebox{0.7}{ 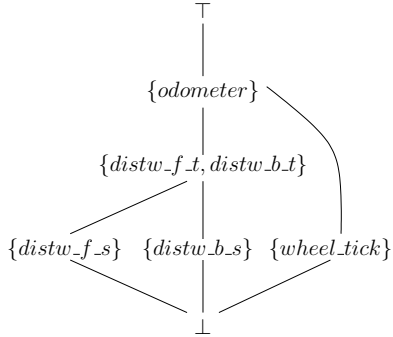 }
\caption{Security lattice representation of the maximal implementation $I_1$.}
\vspace{-0.6cm}
\label{fig:sl}
\end{figure}
\end{example}

\section{Security Lattices For Information-Flow Interfaces}

In this section, we introduce security lattice contracts and show how to translate from flow relations to security lattices and vice-versa.
As for information-flow contracts, we define security lattice contracts over disjoint sets of input and output variables, denoted \(\InVars\) and \(\OutVars\) respectively, with the set of all variables defined as \(\AllVars = \InVars \cup \OutVars\).
Note that while our atomic objects are variables, 
security labels are defined by sets of variables.
For a set of variables \(\AllVars\), we define the set with all possible security labels over \(\AllVars\) as \(\SC^{\AllVars} = 2^{\AllVars}  \cup \{\top, \bot\}\).
A variable \(x\) is labelled with \(\Label \mathop{\in} \SC\) iff \(x \mathop{\in} \Label\).

\begin{definition}
A \emph{security lattice contract} is a tuple \((\InVars, \OutVars, \ContA, \ContG)\) where \(\InVars\) and \(\OutVars\) are disjoint sets of variables and \(\AllVars=\InVars\cup\OutVars\); \(\ContA = (\SC_{\ContA}, \canflow_{\ContA})\) is a security lattice with \(\SC_{\ContA} \subseteq \SC^{\AllVars}\),  called contract \emph{assumption}; and  \(\ContG = (\SC_{\ContG}, \canflow_{\ContG})\) is a security lattice with \(\SC_{\ContG} \subseteq \SC^{\AllVars}\), called contract \emph{guarantee}.
\end{definition}

\subsubsection{From Flow Relations to Security Lattices and Back.}

While security lattices specify requirements on security labels that are later assigned to objects of our design (for example, variables), flow relations specify flow requirements directly between the objects of our design.
For this reason, flow relations should not be anti-symmetric, as we want to keep the design objects distinct even if they are indistinguishable to this relation.

We split the translation from a flow relation \(\Flows\) to security lattice in three steps: (1) defining the set of security labels; (2) defining a partial-order over that set preserving the flow requirements from \(\Flows\); and (3) extending the set of security labels and the can-flow relation to guarantee there is a least-upper bound for each pair of labels.
The algorithm for this translation is presented in Algorithm \ref{alg:tolattice}.
We explain next the steps in detail.

\begin{algorithm}
\caption{From flow relation \(\Flows\) to security lattice -- \(\toLattice{\Flows}\)}
\label{alg:tolattice}
\begin{algorithmic}
\Require \emph{Flow relation} \(\Flows \subseteq (\FlowsFrom \cup \FlowsTo) \times \FlowsTo\)
\Ensure \emph{Security Lattice} \((\SC,\canflow)\)
\begin{algorithmic}[1]
 \State \(\SC \gets \{\SCTop, \SCBot\}\)
\State \(\canflow \gets \{\}\)
 \ForAll{sequences \((z_0,z_1) \circ \cdots \circ (z_n,z_0)\) defined by pairs in \(\Flows\)}
    \State \(\Label_1 \gets \{z_0, \ldots, z_n\}\)
    \State \(\mathrm{isToInclude} \gets \mathrm{true}\)
    \ForAll{\(\Label_2 \mathop{\in}\SC\)}
        \If{\(\Label_1 \mathop{\subseteq} \Label_2\)}
            \State 
            \(\mathrm{isToInclude} \gets \mathrm{false}\) \Comment{We do not add \(\Label_1\) to \(\SC\) as \(\Label_2\) subsumes it}
        \ElsIf{\(\Label_2 \mathop{\subseteq}\Label_1\)}
            \State 
            \(\SC \gets \SC \setminus \{\Label_2\}\)
            \Comment{Remove labels from \(\SC\) subsumed by \(\Label_1\)}
        \EndIf
    \EndFor
    \If{\(\mathrm{isToInclude}\)}
        \(\SC \gets \SC \cup \{\Label_1\}\)
    \EndIf
 \EndFor
 \ForAll{\(\Label_1 \mathop{\in} \SC\)}
    \State \(\canflow \gets \canflow \mathop{\cup}\{(\bot,\Label_1),\ (\Label_1,\top),\ (\Label_1,\Label_1)\}\) 
    \Comment{Label relates with \(\top\), \(\bot\) and itself}
    \ForAll{\(\Label_2 \mathop{\in} \SC\)}
        \If {\(\Label_1 \times \Label_2 \mathop{\subseteq \Flows}\)}
            \State 
            \(\canflow \gets \canflow \mathop{\cup}\{(\Label_1,\Label_2)\}\)
            \Comment{Add pair of labels to can-flow, if they satisfy \(\Flows\)}
        \EndIf
    \EndFor
 \EndFor
 \If{\((\SC, \canflow)\) is not a lattice}
    \State \((\SC, \canflow) \gets \auxClasses((\SC,\canflow))\)
    \Comment{Complete order}
 \EndIf
\State \Return \(\SC, \canflow)\)
\end{algorithmic}
\end{algorithmic}
\end{algorithm}

\begin{algorithm}
\caption{Add least upper bound labels -- \(\auxClasses\)}
\label{alg:addAux}
\begin{algorithmic}
\Require \((\SC,\canflow)\)
\Ensure Updated \((\SC,\canflow)\)
\begin{algorithmic}[1]
 \State \(\mathrm{ToProcess} \gets \{(\Label_1, \Label_2) \ |\ \{\Label_1, \Label_2\} \subseteq \SC\}\) \Comment{Labels to process}
 \ForAll{\((\Label_1, \Label_2) \in  \mathrm{ToProcess}\)}
    \If {\(\Label_1\) and \(\Label_2\) has no least upper bound}
        \State \(\Label' = \{\mathrm{newVar}(\SC)\}\)
        \Comment{Singleton label with new variable name}
        \State \(\mathrm{ToProcess} \gets \mathrm{ToProcess} \cup \{(\Label', \Label) \ |\ \Label \in \SC\}\)
        \State \(\SC \gets \SC \cup \{\Label'\}\) 
        \State \(\canflow \gets \canflow \cup \{(\bot,\Label'),\ (\Label',\top),\ (\Label',\Label'),\ (\Label_1,\Label'),  (\Label_2,\Label')\}\) 
        \ForAll{\(\Label \in \SC\)}
            \If{\(\{(\Label_1,\Label), (\Label_2,\Label)\} \subseteq \canflow \)}
                \State \(\canflow \gets \canflow \cup \{(\Label', \Label)\}\)
                \Comment{Connect \(\Label'\) to all common upper bounds}
            \EndIf  
        \EndFor
    \EndIf
    \State \(\mathrm{ToProcess}  \gets \mathrm{ToProcess}  \setminus \{(\Label_1, \Label_2)\}\) \Comment{Remove processed labels}
 \EndFor
\State \Return \((\SC,\canflow)\)
\end{algorithmic}
\end{algorithmic}
\end{algorithm}

\paragraph{Security Labels.} 
In our first step, we collect all sets of variables in which all variables are in a loop defined using pairs in the given flow relation (c.f., line 3 in Algorithm \ref{alg:tolattice}).
Formally, for a flow relation \(\Flows\) and for a loop defined by a sequence \((z_0,z_1) \circ \cdots \circ (z_n,z_0)\), where \((z_i, z_{i+1}) \mathop{\in} \Flows\) for \(0\leq i<n\), then \(\{z_0, \ldots, z_{n}\}\) defines a security label.
Note that as flow relations are reflexive, for each variable \(z\) in the domain of \(\Flows\), the singleton set \(\{z\}\) defines a security label that satisfies the loop condition.
We observe, additionally, that due to transitivity of flow relations all subsets of labels defined as explained above will also be security labels according to this definition (i.e., with transitivity we can skip steps in the sequence and find a loop that supports the label with fewer variables).
To remove redundant labels, we keep only the largest set of variables  necessary to define the loop (c.f., lines 7 to 13 in Algorithm \ref{alg:tolattice}).

\paragraph{Can-flow Order.}
After creating the labels, we proceed to defining the can-flow relation as a partial order over the set of labels derived in the previous step.
This is straightforward, as the flow relation is already reflexive and transitive.
A pair of security labels \((\Label, \Label')\) is in the can-flow relation iff all the variables in them are related by the flow relation (c.f., lines 16 to 23 in Algorithm \ref{alg:tolattice}).
Formally, given a flow relation \(\Flows\), for all \(z\mathop{\in} \Label\) and \(z'\mathop{\in}\Label'\), we require \((z,z') \mathop{\in}\Flows\).
Additionally, we make sure that all labels are appropriately connected to the top and bottom elements of the lattice  and to themselves to guarantee reflexivity (c.f., line 17 in Algorithm~\ref{alg:tolattice}).

\begin{example}
In Figure \ref{fig:inc}, we depict an implementation of the shared communication infrastructure, \(I_3\).
To the right of that implementation, we depict the partial order \((\SC,\canflow)\) derived by line 23 of the Algorithm \ref{alg:tolattice}.
Note that the pair of labels \((\{\dwFs\}, \{\dwBs\})\) do not have a lower upper bound, because \(\{\dwFt\}\) and \(\{\dwBt\}\) are uncomparable for the relation \(\canflow\).
When this is the case, we proceed in Algorithm \ref{alg:tolattice} to line 25 and add the missing labels. 
\begin{figure}
\centering
\scalebox{0.7}{ 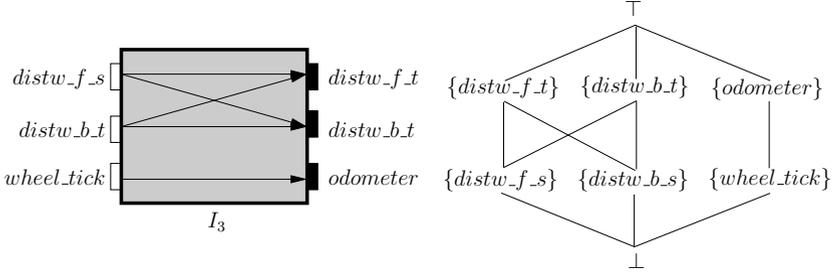 }
\caption{An implementation for the shared communication \(I_3\) with the partial order derived from its can-flow by line 23 in Algorithm \ref{alg:tolattice}. Note that labels \(\{\dwFs\}\) and \(\{\dwBs\}\) do not have a lower upper bound, which will be added in the next steps of the Algorithm.}
\vspace{-0.6cm}
\label{fig:inc}
\end{figure}
\end{example}

\paragraph{Least upper bounds.}
As illustrated in the example above, the set of labels together with the can-flow relation by line 23 in Algorithm \ref{alg:tolattice} may not define a lattice.
In Algorithm \ref{alg:addAux}, we present a possible way to extend the set of security labels and the can-flow relation to guarantee that each pair of labels has a unique least upper bound.
There are multiple ways one could define and implement this extension.
In this work, we focus on presenting a simple and intuitive extension for demonstration purposes.

Intuitively, for a flow relation over \((\FlowsFrom \cup \FlowsTo) \times \FlowsTo\), we need to create intermediate labels connecting labels with \(\FlowsFrom\) variables to labels with \(\FlowsTo\) variables.
Note that there are no labels that mix variables of \(\FlowsTo\) and \(\FlowsFrom\) because there is no loop containing both types of variables in the flow relation \(\Flows \subseteq (\FlowsFrom \cup \FlowsTo) \times \FlowsTo\).
To define the intermediate labels, we extend the domain of security labels with sets over fresh variables (obtained with the function \(\mathrm{newVar}\) that returns a variable name not yet in \(\SC\)).
The new labels are internal to a contract implementation (i.e., irrelevant to the system's design) and characterize the expected security label when we combine information from different sources.
\\

\begin{example}
Algorithm \ref{alg:tolattice} outputs lattices in Figure \ref{fig:sl} and Figure \ref{fig:complete}
when the input are the flow relations in implementations \(I_1\) and \(I_3\), respectively.
\begin{figure}
\centering
\scalebox{0.7}{ 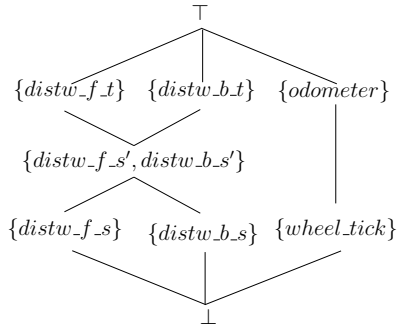 }
\caption{Security lattice returned from Algorithm \ref{alg:tolattice} for the flow relation in \(I_3\).}
\vspace{-0.6cm}
\label{fig:complete}
\end{figure}
\end{example}

We prove below that the relation output by Algorithm \ref{alg:tolattice} is indeed a lattice and is equivalent to the flow restrictions of the input \(\Flows\).
For a flow relation \(\Flows \subseteq \AllVars \times \FlowsTo\),  a structure \((\SC,\canflow)\) \emph{is equivalent to it}, denoted 
\((\SC,\canflow) \equiv_{\AllVars\times \FlowsTo} \Flows\) iff for all pairs of labels \((\Label_1, \Label_2) \in 2^\AllVars  \times 2^\FlowsTo\), they are in the can-flow relation, \((\Label_1, \Label_2) \in \canflow\), iff,
for all \((z_1,z_2) \in \Label_1 \times \Label_2\) we have \((z_1,z_2) \in \Flows\).

\begin{theorem}
\label{thm:toLattice:correctness}
Let \(\Flows \subseteq \AllVars \times \FlowsTo\) be a flow relation  and
\(\toLattice{\Flows} = (\SC,\canflow)\).
Then, \((\SC,\canflow)\) defines a security lattice that is equivalent to the can-flow restrictions of \(\Flows\) over \(\AllVars \times \FlowsTo\), i.e.,
\((\SC,\canflow) \equiv_{\AllVars\times \FlowsTo} \Flows\).
\end{theorem}

\begin{proof}
We start by proving that \((\SC,\canflow)\) defined at line 23 of Algorithm~\ref{alg:tolattice} (before we call \(\auxClasses\) procedure) 
(i) defines an order that is equivalent to \(\Flows\) over its domain, \((\SC,\canflow) \equiv_{\AllVars\times \FlowsTo} \Flows\), and
(ii) \(\canflow\) defines a partial order over the set of labels \(\SC\).

We observe that in Algorithm~\ref{alg:tolattice} (a) by line 16, we only have labels with variables in the domain of \(\Flows\) (i.e., \(\SC \subseteq 2^{\FlowsFrom \cup \FlowsTo} \cup \{\SCTop, \SCBot\}\));
(b) only in lines 17 and 20 we add pairs of labels to \(\canflow\); and (c) in line 17 we add the reflexive pairs and the connections to \(\SCTop\) and \(\SCBot\), while, by the condition at line 19, in line 20 we only add the pair \((\Label_1,\Label_2)\) to the can-flow relation when all their elements \(z_1 \mathop{\in} \Label_1\) and 
\(z_2 \mathop{\in} \Label_2\) are pairs in the flow relation, \((z_1,z_2) \in \Flows\).
By these three observations, it follows that by line 23 of Algorithm~\ref{alg:tolattice} \((\SC,\canflow) \equiv_{\AllVars\times \FlowsTo} \Flows\) holds.
We prove now (ii).
After line 17 in Algorithm~\ref{alg:tolattice}, \(\canflow\) is a reflexive relation.
By transitivity of flow relations and the point (i) proved above, then \(\canflow\) is also a transitive relation.
Finally, by removal of the loops in lines 3 to 15, it follows that \(\canflow\) is antisymmetric.

Now, we are only missing to prove that by line 27  in Algorithm \ref{alg:tolattice} the returned \((\SC,\canflow)\)
(iii) is equivalent to \(\Flows\),\((\SC,\canflow) \equiv_{\AllVars\times \FlowsTo} \Flows\), and
(iv) defines a lattice.
First, we observe that all pairs of labels added to \(\canflow\) in lines 7 and 10 of Algorithm \ref{alg:addAux} do not have variables from \(\AllVars\).
Then, by (i) proved previously, (iii) holds.
If \((\SC,\canflow)\) is a lattice by line 24 Algorithm \ref{alg:tolattice}, then \((\SC,\canflow)\) is not changed so (iv) holds.
We consider now the case that \((\SC,\canflow)\) does not define a lattice by line 24.
Due to the elements added to \(\canflow\) in line 7 of Algorithm \ref{alg:addAux} and (ii) proved above, it follows that \(\canflow\) is reflexive relation.
Moreover, after the same line and from line 17 of Algorithm \ref{alg:tolattice}, it follows that \(\SCBot\) and \(\SCTop\) are the lower and upper bound of \((\SC, \canflow)\), respectively.
Transitivity of \(\canflow\) follows from the pairs added in line 7 and 10 in Algorithm \ref{alg:tolattice} and (ii) proved above.
By lines 1 and 5 in Algorithm \ref{alg:addAux}, we have the guarantee that all pair of labels added to \(\SC\) will be checked for having a least upper bound.
If they do not have it, a new label is created that is connected to all their upper bounds (c.f., line 10).
As all the upper bounds to any of pair of labels processed are either in the algorithm input structure or are added only during the loop iteration we add a new label, then we have the guarantee that no new upper bounds are added after a pair of labels is processed.
Then, at each iteration the new label \(\Label'\) is the least upper bound for the labels \(\Label_1\) and \(\Label_2\).
Hence the returned \((\SC,\canflow)\) is a lattice.
\end{proof}

The translation from a security lattice to a flow relation is straightforward, as a lattice is by definition already reflexive and transitively closed.
Formally, given a security lattice \((\SC,\canflow)\) we define its derived flow relation over a domain \(\AllVars \times \FlowsTo\) as follows:
\[
\toFlowRel{(\SC,\canflow), \AllVars \times \FlowsTo} = 
\{(z,z') \ |\ (\Label,\Label') \mathop{\in}\canflow,\  z\mathop{\in} (\Label\mathop{\cap} \AllVars) \tAnd z'\mathop{\in} (\Label'\mathop{\cap} \FlowsTo)\}.
\]

We prove below that after we translate a flow relation to a security lattice, using the operator defined by Algorithm \ref{alg:tolattice}, we can use the operator defined above to get back the original flow relation.

\begin{theorem}
\label{th:fromFlowsBackToFlows}
Let \(\Flows\subseteq (\FlowsFrom \cup \FlowsTo) \times \FlowsTo\) be a flow relation with \(\AllVars = \FlowsFrom \cup \FlowsTo\).
Then:
\[\toFlowRel{\toLattice{\Flows}, \AllVars \times \FlowsTo} = \Flows.\]
\end{theorem}

\begin{proof}
By Theorem \ref{thm:toLattice:correctness},  
\(\toLattice{\Flows} \equiv_{\AllVars \times \FlowsTo} \Flows\), i.e., for \((\SC,\canflow) = \toLattice{\Flows}\) and
 for all pairs of labels \((\Label_1, \Label_2) \in 2^\AllVars  \times 2^\FlowsTo\), they are in the can-flow relation, \((\Label_1, \Label_2) \in \canflow\), iff,
for all \((z_1,z_2) \in \Label_1 \times \Label_2\), we have \((z_1,z_2) \in \Flows\).
Hence, by definition of \(\toFlowRel{.}\), \(\toFlowRel{\toLattice{\Flows}, \AllVars \times \FlowsTo} = \Flows\) holds.
\end{proof}

We show in the theorem below that, under reasonable constraints, when given a security lattice, if we translate it to a flow relation followed by a translation to a lattice, then at the end, we get back the original lattice.
The first constraint, specifying that variables in \(\FlowsFrom\) can only define singleton labels, is motivated by the observation that for a flow relation \(\Flows \subseteq (\FlowsFrom \cup \FlowsTo) \times \FlowsTo\), there are only outgoing flows from \(\FlowsFrom\). 
Then, there cannot be loops with variables in \(\FlowsFrom\), and in our translation from flow relations to lattices, we defined that loops are the only way to get non-singleton labels.
The second constraint is derived from our implementation decision in Algorithm \ref{alg:tolattice} only to keep maximal sets.

\begin{theorem}
\label{th:fromLatBackToLat}
Let 
\(\SC \subseteq \{\{u\} \ |\ u\in \FlowsFrom\} \cup 2^{\FlowsTo} \cup \{\SCTop,\SCBot\}\) be a set of maximal sets of variables, i.e., for all \(\{\Label_1, \Label_2\} \subseteq \SC\) then \(\Label_1 \not\subseteq \Label_2\).
For all relations \(\canflow \subseteq \SC \times \SC\) s.t.\ 
\((\SC,\canflow)\) defines a lattice, then:
\[\toLattice{\toFlowRel{(\SC,\canflow), (\FlowsFrom \cup \FlowsTo) \times \FlowsTo}} = (\SC,\canflow).\]
\end{theorem}
\begin{proof}
Let \(\Flows = \toFlowRel{(\SC,\canflow), (\FlowsFrom \cup \FlowsTo) \times \FlowsTo}\).
Note that, by \(\canflow\) being a partial order over \(\SC\) and by definition of \(\toFlowRel{.}\), 
the only loops in the flow relation output by \(\toFlowRel{(\SC,\canflow), (\FlowsFrom \cup \FlowsTo) \times \FlowsTo}\)
are between variables in the same label in \(\SC\).
Let \(\SC_{15}\) be the set of security labels by line 15 in  Algorithm \ref{alg:tolattice}.
By \(\SC\) having only maximal sets, then \(\SC_{15} = \SC\).
Now, let \(\canflow_{23}\) be the can-flow relation defined by line 23 of Algorithm \ref{alg:tolattice} and \(\SC_{23}\) the set of labels (note that \(\SC_{15} = \SC_{23} = \SC\)).
By a reasoning analogous to the proof of  Theorem~\ref{thm:toLattice:correctness}, we know that 
\((\SC_{23},\canflow_{23}) \equiv_{\AllVars\times \FlowsTo} \Flows\).
As \(\Flows\) is derived from a lattice, then \((\SC_{23},\canflow_{23})\) is also a lattice and \(\toLattice{\toFlowRel{(\SC,\canflow), (\FlowsFrom \cup \FlowsTo) \times \FlowsTo}} = (\SC_{23},\canflow_{23}) = (\SC, \canflow_{23})\).
Then, by \((\SC_{23},\canflow_{23}) \equiv_{\AllVars\times \FlowsTo} \Flows\), it follows that \(\canflow_{23} = \canflow\).
Hence \(\toLattice{\toFlowRel{(\SC,\canflow), (\FlowsFrom \cup \FlowsTo) \times \FlowsTo}} = (\SC_{23},\canflow_{23}) = (\SC,\canflow)\)
\end{proof}

\subsubsection{Information-flow and Security Lattice Contracts.}
Given an information-flow contract \(C_{\fl} = (\InVars, \OutVars, \AFlow, \GFlow)\), we define the derived security lattice contract as:
\(\toLatC{C_{\fl}} = (\InVars, \OutVars, \toLattice{\AFlow}, \toLattice{\GFlow})\).
Information-flow contracts were introduced in \cite{IFInterfaces20Arxiv} along with a composition operator and a refinement relation over such contracts.
Using the translations presented above, we define composition and refinement over security lattice contracts directly from their counterpart in information-flow contracts (we use the same definitions and translate back and forth between flows and lattices).

In \cite{IFInterfaces20Arxiv}, we proved that the composition and refinement of information-flow contracts adequately capture the semantics of their counterpart in information-flow interfaces.
These results show that information-flow contracts, just as for information-flow interfaces, satisfy incremental design and independent implementability \cite{IFInterfaces20Arxiv}.
Then, using Theorem \ref{th:fromFlowsBackToFlows}, all the results mentioned above for information-flow contracts also hold for security lattices contracts.

\section{Conclusions}
\label{sec:conc}

In this paper, we proposed security lattices as an alternative semantics for the information-flow interface theory. 
Security lattices are the backbone of many successful, readily available techniques to verify and enforce security policies.
Examples range from type systems \cite{focardi2011types} to static analysis \cite{huangStaticAnalysisLat04} or runtime enforcement using secure multi-execution \cite{groefFlowfox12,austin2012multiple}, to name a few.

Equipping interfaces with a semantic interpretation based on a commonly used formalism for defining and enforcing security policies facilitates the bridge between the security practitioners and the security-by-design paradigm.

\subsubsection{Acknowledgements} This project was funded in part by the Austrian Science Fund (FWF) SFB project SpyCoDe F8502 and by the ERC-2020-AdG 101020093.

\bibliographystyle{splncs04}
\bibliography{main}

\clearpage
\appendix

\end{document}